\documentclass{amsart}

\usepackage{amsmath}
\usepackage{amssymb}
\usepackage{amsfonts}
\usepackage{amsthm}

\newtheorem{lemma}{Lemma}[section]
\newtheorem{theorem}{Theorem}

\theoremstyle{definition}

\theoremstyle{definition}

\theoremstyle{definition}

{\catcode `\@=11 \global\let\AddToReset=\@addtoreset}
\AddToReset{equation}{section}

\newcommand{\cL}{{\mathcal L}}

\newcommand{\rd}{{\rm d}}
\newcommand{\ph}{\varphi}

\newcommand{\bx}{\mathbf{x}}

\newcommand{\cE}{{\mathcal E}}

\newcommand{\wt}{\widetilde}

\newcommand{\cH}{\mathcal{H}}

\newcommand{\bR}{{\mathbb R}}

\def\tr{\mathop{\rm tr}\nolimits} 

\title[Finite time blowup for the Hartree-Fock equation]{Stellar Collapse in the Time Dependent Hartree-Fock Approximation}

\begin{document}

\author{Christian Hainzl}
\address{Departments
  of Mathematics and Physics, UAB, 
Birmingham AL 35294, USA} \email{hainzl@math.uab.edu}

\author{Benjamin Schlein}
\address{Institute of Mathematics, University of Munich,
Theresienstr. 39, D-80333 Munich, Germany}
\email{schlein@mathematik.uni-muenchen.de}

\begin{abstract}
We prove blow-up in finite time for radially symmetric solutions to the pseudo-relativistic Hartree-Fock equation with negative energy. The non-linear Hartree-Fock equation is commonly used in the physics literature to describe the dynamics of white dwarfs. We extend thereby recent results by Fr\"ohlich and Lenzmann, who established in \cite{FL1,FL2} blow-up for solutions to the pseudo-relativistic Hartree equation. As key ingredient for handling the exchange term we use the conservation of the expectation of the square of the angular momentum operator.
\end{abstract}

\maketitle

\section{Introduction}

According to the theory of Chandrasekhar \cite{chandra}, white dwarfs can be described by a model of electrically neutral atoms interacting through classical Newtonian gravitation. Atoms consist of nuclei, which are responsible for the main part of the potential energy, and electrons, which, on the other hand, give the leading contribution to the kinetic energy of the star. Because of local charge neutrality, we can assume that the space and momentum distributions of the nuclei coincide with the ones of the electrons; in this approximation, we only keep track of the electronic degrees of freedom. Considering a relativistic dispersion $E(p) = \sqrt{p^2 +m^2}$ for electrons with mass $m$, and assuming a single species of nuclei with mass $m_Z \gg m$ and charge $Ze$ (where $-e$ denotes the charge of the electron), this simplified model is described, on the microscopic level, by the quantum mechanical Hamiltonian
\begin{equation}\label{eq:ham} H_N = \sum_{j=1}^N \sqrt{-\Delta_{x_j} + m^2} -
\kappa \sum_{i<j}^N \frac{1}{|x_i - x_j|}\, \end{equation} where $N$ is the number of electrons and $\kappa = G m_Z^2/Z^2$ ($G$ denotes here the gravitational constant). We use units with $\hbar=c=1$.

Since the electron spin does not play an important role, we neglect it, and describe electrons by wave functions in $L^2 (\bR^3)$. In accordance with Pauli's principle, the Hamiltonian \eqref{eq:ham} acts then on the antisymmetric tensor product space $\cH_N = \bigwedge^N  \, L^2 (\bR^3, \rd x)$.

In \cite{LY1,LY2} Lieb and Yau derived Chandrasekhar equation for the ground state of the Hamiltonian \eqref{eq:ham} in the limit of large $N$ with $\kappa N^{2/3}$ kept fixed. Additionally they reproduced Chandrasekhar critical number of particles $N_c \sim O(\kappa^{-3/2})$ proving the instability of \eqref{eq:ham} for all $N > N_c$ (the white dwarf is supposed to undergo gravitational collapse for $N > N_c$). Up to a factor $4$ the correct value of $N_c$ had already been established by Lieb and Thirring in \cite{LT}. See \cite{LY1,LY2,LT,FL2} for a more thorough discussion on white dwarfs. In principle \eqref{eq:ham} can also be used to describe neutron stars; however, in this case, a correct understanding of the collapse requires the inclusion of general relativity effects.

In the physical relevant regime of very small $\kappa$ and very large $N$, one expects the ground state of (\ref{eq:ham}) to be approximated by a Slater determinant
\[ \left( \psi_1 \wedge \psi_2 \wedge \dots \wedge \psi_N \right)
(\bx) = \frac{1}{\sqrt{N!}} \sum_{ \pi \in S_N} \sigma_{\pi} \psi_1
(x_{\pi 1}) \psi_2 (x_{\pi 2}) \dots \psi_N (x_{\pi N}) \] of $N$
orthonormal one-particle wave functions $\psi_j \in L^2 (\bR^3)$
(here the sum runs over all permutations of the $N$ particles;
moreover, $\sigma_{\pi} = 1$ if the permutation $\pi$ is even while
$\sigma_{\pi} = -1$ if it is odd). It is simple to verify that the
energy of the Slater determinant $\wedge_{j=1}^N \psi_j$ is given by
the so called Hartree-Fock functional
\begin{equation}\label{eq:HFen}
\begin{split}
\cE_{{\rm HF}} (\{ \psi_j \}_{j=1}^N) = \; &\sum_{j=1}^N \int \rd x
\, |(-\Delta +m^2)^{1/2} \, \psi_j (x)|^2 \\ &- \frac{\kappa}{2}
\sum_{i,j=1}^N \int \rd x \, \rd y \, \frac{|\psi_i (x)|^2 |\psi_j
(y)|^2 - \psi_i (x) \overline{\psi}_i (y)  \psi_j (y)
\overline{\psi}_j (x)}{|x-y|} \, .
\end{split}
\end{equation}

\medskip

Within the range of its applicability,
one also expects the Hartree-Fock theory to describe the
time-evolution of Slater determinants. In other words, one expects
that, in an appropriate sense, and in a suitable limit of large $N$
and small $\kappa$, \[ e^{-iH_N t} \left( \psi_1 \wedge \dots \wedge
\psi_N \right) \simeq \psi_1 (t) \wedge \psi_2 (t) \wedge \dots
\wedge \psi_N (t) \] where the wave functions evolve according to
the time-dependent Hartree-Fock equation
\begin{equation}\label{eq:HF1}
i\partial_t \psi_{j} = \sqrt{-\Delta + m^2} \, \psi_{j} - \kappa \sum_{i=1}^N \left( \frac{1}{|.|} * |\psi_{i}|^2 \right) \psi_{j} + \kappa \sum_{i=1}^N \left(\frac{1}{|.|} * \psi_{j} \overline{\psi}_{i} \right) \psi_{i}\,.
\end{equation}
Note that this system of non-linear equations, which can be formally obtained computing the variation of (\ref{eq:HFen}), preserves the orthonormality relations $\langle \psi_i (t), \psi_j (t) \rangle = \delta_{ij}$ and the energy $\cE_{\text{HF}}$.

\medskip

The Hartree-Fock theory can be formulated in a more compact form in terms of the orthogonal projection $Q = \sum_{j=1}^N |\psi_j \rangle \langle \psi_j|$ onto the subspace spanned by the wave functions $\{ \psi_j \}_{j=1}^N$. The Hartree-Fock energy (\ref{eq:HFen}) is given, in terms of $Q$ and its kernel $Q(x,y)$, by
\begin{equation}\label{eq:HFen2}
\cE_{{\rm HF}} (Q) = \tr \; \sqrt{-\Delta +m^2} \, Q -
\frac{\kappa}{2} \int \rd x \rd y \, \frac{Q(x,x) Q(y,y) -
|Q(x,y)|^2}{|x-y|} \,.
\end{equation}
Also the time-dependent Hartree-Fock system (\ref{eq:HF1}) can be translated into an evolution equation for the time dependent density $Q_t = \sum_{j=1}^N |\psi_j (t) \rangle \langle \psi_j (t)|$. It is easy to obtain the nonlinear Hartree-Fock equation
\begin{equation}\label{eq:HF2}
i\partial_t Q_t = \left[ (-\Delta+ m^2)^{1/2} - \kappa \left( \frac{1}{|.|} * \rho_{Q_t} \right) + \kappa R_{Q_t} , Q_t \right] \, ,
\end{equation}
where $\rho_{Q_t} (x) = Q_t (x,x)$ (we denote by $Q_t (x,y)$ the
kernel of the projection $Q$), and where the operator $R_{Q_t}$ is
defined by its kernel $R_{Q_t} (x,y) = Q_t (x,y)/|x-y|$. By
construction, it is clear that (\ref{eq:HF2}) preserves the trace $N
= \tr Q_t$, and the energy (\ref{eq:HFen2}). In the present paper we
are interested in solutions to the nonlinear Hartree-Fock equation
(\ref{eq:HF2}); in particular we prove the existence of solutions to
(\ref{eq:HF2}) which exhibit blow up in finite time. Within the
framework of the Hartree-Fock approximation, the blow up of
solutions to (\ref{eq:HF2}) is interpreted as evidence for the
dynamical collapse of  white dwarfs.

\medskip

The last contribution in the commutator on the r.h.s. of (\ref{eq:HF2}) (the term containing the operator $R_Q$) is known as the exchange term (while the second contribution, containing the density $\rho_Q$, is known as the direct term). The presence of the exchange term is a consequence of the Pauli principle. Since, in the relevant limit of large $N$ and small $\kappa$, the exchange term is expected to be of smaller order compared with the direct term, it is often neglected in the physics literature. In this approximation one obtains the Hartree equation
\begin{equation}\label{eq:hartree} i\partial_t Q_t = \left[ (-\Delta+ m^2)^{1/2} - \kappa \left( \frac{1}{|.|} * \rho_{Q_t} \right), Q_t \right] \,.
\end{equation}
For bosonic systems (boson stars) this equation has been in fact rigorously derived from many body quantum dynamics in \cite{ES}.

Recently, blow-up in finite time has been proven to occur for solutions to the Hartree equation (\ref{eq:hartree}) by Fr\"ohlich and Lenzmann in \cite{FL1,FL2}. To obtain this result, they consider the non-negative observable \[ M = x \sqrt{-\Delta + m^2} \, x = \sum_{j=1}^3 x_j \sqrt{-\Delta +m^2}\, x_j \]  and they estimate the expectation value $\tr (M Q_t)$ where $Q_t$ is a solution to (\ref{eq:hartree}) with spherical symmetry (in the sense that $Q_t (Rx,Ry) = Q_t (x,y)$ for all $R \in SO(3)$). Under this assumption, they show that \begin{equation}\label{eq:obs0} \tr M Q_t \leq 2 t^2 \cE_{\text{Hartree}} (Q) + O(t) \end{equation} where $O(t)$ denotes error terms growing at most linearly in $t$. For initial data with negative energy $\cE_{\text{Hartree}} (Q) <0$, Eq. (\ref{eq:obs0}) leads to a contradiction to the non-negativity of the observable $M$ (choosing $t$ sufficiently large). This implies that the solution $Q_t$ cannot exist globally in time.

The spherical symmetry of the density $Q$ plays a very important role in the analysis developed by Fr\"ohlich and Lenzmann; it allows them to control error terms arising from the commutator of $M$ with the interaction $(|.|^{-1} * \rho_{Q_t})$ (the time derivative of $\tr \; M Q_t$ contains the term $\tr \; [M, (|.|^{-1}*\rho_{Q_t})] Q_t$). The same approach can be applied to solutions of the Hartree-Fock equation (\ref{eq:HF2}); it turns out, however, that the error terms arising from the commutator of $M$ with the exchange term $R_Q$ cannot be handled like the errors arising from the direct term. This is the reason why the approach of Fr\"ohlich and Lenzmann does not extend in a simple way to the Hartree-Fock equation (\ref{eq:HF2}) (the method does extend to the Hartree-Fock equation if one assumes not only that $Q$ is spherically symmetric, but also that each orbital in its decomposition $Q=\sum_{j=1}^N |\psi_j\rangle \langle \psi_j|$ is spherically symmetric; however, as pointed out in \cite{FL2}, this is a physically unnatural condition).

\medskip

In our analysis, we use the strategy of Fr\"ohlich and Lenzmann; we
study the evolution of the observable $M = x \, \sqrt{-\Delta + m^2}
\, x$ on spherically symmetric solutions $Q_t$ of (\ref{eq:HF2})
and we obtain a bound like (\ref{eq:obs0}). The novelty of our
approach lies in the estimate of the error terms arising from the
commutator of $M$ with the exchange term $R_Q$; to control this
error terms, we make use of the expectation of the square of the angular
momentum operator, which is a conserved quantity due to the radial
symmetry.

\section{The main result and its  proof}

The local well-posedness of the Hartree-Fock system (\ref{eq:HF1}) has been established by Fr\"ohlich and Lenzmann in \cite[Theorem 1]{FL2}. The well-posedness of (\ref{eq:HF2}) follows along the same line. For $s \geq 0$ we define the space
\[ \cH_s = \{ Q \in \cL^1 (L^2 (\bR^3)): \| Q \|_{\cH_s} < \infty \} \] with the norm
\[ \| Q \|_{\cH_s} = \tr \; \left| \, (1 - \Delta)^{s/2} Q (1-\Delta)^{s/2} \, \right| \, . \]
It turns out that (\ref{eq:HF2}) is locally well-posed in $\cH_s$ for all $s \geq 1/2$.
\begin{theorem}[Local Well-Posedness, \cite{FL2}]\label{thm:loc} Fix $s \geq 1/2$. For every orthogonal projection $Q\in \cH_s (L^2 (\bR^3))$ with $N (Q) = \tr \, Q < \infty$ there exists a maximal existence time $T >0$ and a unique solution $Q_t \in C ([0,T), \cH^s (L^2 (\bR^3)))$ of (\ref{eq:HF2}) with $Q_t \geq 0$, and $\tr Q_t = N$ for all $t \in [0,T)$. Moreover, if $T < \infty$, then $\| Q_t \|_s = \tr (1-\Delta)^s Q_t \to \infty$ as $t \to T^-$ (blow-up alternative).
\end{theorem}

\medskip

For sufficiently small values of $N=\tr \; Q$, Fr\"ohlich and Lenzmann also proved
global well-posedness of the Hartree-Fock equation in \cite[Theorem 2]{FL2}. Our goal here is to prove that, for sufficiently large values of $N = \tr \; Q$, blow up in finite time can occur. In particular, we show that the time evolution of an arbitrary spherically symmetric density with negative energy and with finite expectation for the square of the angular momentum operator exhibits blow-up in finite time. To prove this result, we need a few simple preliminary lemmas. First of all, we need to prove that the spherical symmetry is preserved by the time-evolution; this follows easily from the local uniqueness of the solution to (\ref{eq:HF2}).
\begin{lemma}\label{lm:spher}
Let $Q\in \cH^{1/2} (L^2 (\bR^3))$, $Q \geq 0$, and assume that $Q$ is spherically symmetric in the sense that
\[ Q (Rx,Ry) = Q (x,y) \qquad \text{for all } R \in SO(3). \] For $t \in [0,T)$ denote by $Q_t$ the local in time solution to (\ref{eq:HF2}) with $Q_{t=0} = Q$. Then, for every $t \in [0,T)$, we have
\begin{equation}\label{eq:QR}
Q_t (Rx , Ry) = Q_t (x,y) \qquad \text{for all } R \in SO(3) \, .
\end{equation}
\end{lemma}
\begin{proof}
For arbitrary $R \in SO(3)$, we define $\wt Q_t (x,y) = Q_t (Rx,Ry)$. It is then simple to verify that $\wt Q_t$ is also a solution to (\ref{eq:HF2}), characterized by the same initial data. By the local uniqueness of the solution, we immediately obtain (\ref{eq:QR}).
\end{proof}

The main reason why spherical symmetry is so important to prove the
blow up of solutions of (\ref{eq:HF2}) is Newton's Law, as stated in
the following lemma.
\begin{lemma}[Newton's Law]\label{lm:nl} Suppose that $\rho \in L^1 (\bR^3, \rd x)$ is spherical symmetric with $N = \int \rd y \, \rho (y)$. Then
\begin{equation}\label{eq:nl1}
\int \rd y \, \frac{\rho (y)}{|x-y|} \leq \frac{N}{|x|} \qquad \text{and} \qquad \left| \nabla_x \, \int \rd y \frac{\rho (y)}{|x-y|} \right| \leq \frac{N}{|x|^2}\,
\end{equation}
for a.e. $x \in \bR^3$.
\end{lemma}
\begin{proof}
The proof relies on the explicit formula for radial functions $$\int \rd y \, \frac{\rho (y)}{|x-y|} = \frac 1{|x|}\int_{|y| \leq |x|}\rd y \, \rho(y) + \int_{|y| > |x|}
\rd y \, \frac{\rho(y)}{|y|}.$$ For $\rho \in L^1 (\bR^3) \cap C^0 (\bR^3)$, explicit differentiation leads to
$$ \left| \nabla_x \, \int \rd y \frac{\rho (y)}{|x-y|} \right| = \frac 1{|x|^2}\int_{|y| \leq |x|}\rd y
\, \rho(y)\leq \frac N{|x|^2}.$$ For general spherical symmetric $\rho \in L^1 (\bR^3)$ the statement follows using a simple density argument.
\end{proof}

An important tool in the proof of the finite time blow-up of solutions of (\ref{eq:HF2}) is the fact that the expectation of the square of the angular momentum operator $L = x \wedge p$ is preserved by the time evolution thanks to  spherically symmetric solutions.
\begin{lemma}\label{lm:L2}
Let $L = x \wedge p$ denote the angular momentum operator. Let $Q\in \cH^{1/2} (L^2 (\bR^3))$, $Q \geq 0$ be a spherical symmetric density with \[ \tr\; L^2 \, Q < \infty \, . \] For $t \in [0,T)$, denote by $Q_t$ the local in time solution of (\ref{eq:HF2}) with $Q_{t=0} = Q$. Then \[ \tr \; L^2 \, Q_t = \tr \; L^2 Q \qquad \text{for all } t \in [0,T) . \]
\end{lemma}

\begin{proof}
Observe that the angular momentum operator $L$ generates rotations in the sense that \[ \left( e^{i \, L \cdot \alpha} \psi \right) (x) = \psi (R_{\alpha} x) \] for all $\psi \in L^2 (\bR^3)$. Here $R_{\alpha} \in SO(3)$ denotes the rotation around the axis $\hat \alpha$, with angle $|\alpha|$. This implies that, for an arbitrary spherical symmetric density $Q$, we have
\[ \left( e^{iL\cdot \alpha} Q e^{-iL \cdot \alpha} \right) (x,y) = Q (R_{\alpha} x , R_{\alpha} y) = Q (x,y) \]
for all $\alpha$. Differentiating with respect to $\alpha$, we obtain that $[ L , Q ] = 0$ and thus that $[L^2, Q]= 0$ for any spherically symmetric density $Q$. Now, since the time evolution $Q_t$ of the spherical symmetric initial data $Q$ is spherical symmetric (by Lemma~\ref{lm:spher}), it follows that
\begin{equation*}
\begin{split}
\frac{\rd}{\rd t} \, \tr \; L^2 Q_t  = \; & \tr \; L^2 \left[ \sqrt{p^2 +m^2} -\kappa \left( \frac{1}{|.|} * \rho_{Q_t} \right) + \kappa R_{Q_t} , Q_t \right] \\ = \; & \tr \; \left( \sqrt{p^2 +m^2} -\kappa \left( \frac{1}{|.|} * \rho_{Q_t} \right) + \kappa R_{Q_t}\right)  \left[ Q_t,  L^2 \right]  =  0 \, .
\end{split}
\end{equation*}
\end{proof}

\medskip

We are now ready to state and prove our main theorem.
\begin{theorem}(Blow-up for Hartree-Fock)
\label{thm:main}
Let $Q \in \cH^{1/2} (L^2 (\bR^3))$ be a spherical symmetric orthogonal projection with $N(Q) = \tr Q < \infty$, with negative energy $\cE_{\text{HF}} (Q) < 0$, with finite expectation of the square of the angular momentum operator $\cL^2 (Q) = \tr \, L^2 Q < \infty$, such that
\begin{equation}
\label{eq:tech}
\tr \; x^4 \, Q < \infty, \qquad \text{and } \quad \tr \; (-\Delta) \,Q < \infty \, .
\end{equation}
For $t \in [0,T)$, let $Q_t$ denote the maximal local in time solution to (\ref{eq:HF2}) with $Q_{t=0} = Q$. Then $T < \infty$ and \[ \| Q_t \|^2_{H^{1/2}} = \tr \; (1-\Delta)^{1/2} Q_t \to \infty \qquad \text{as } t \to T^- \, .\]
\end{theorem}
{\it Remarks.}
\begin{itemize}
\item[$\bullet$] According to \cite[Theorem 4]{FL2} our main Theorem \ref{thm:main}  implies that when approaching the time of blow-up  any blow-up solution  exhibits a concentration of particles at the origin, with
    \[ \lim\inf_{t\to T^-} \int_{|x| \leq R} \rd x\, \rho_{Q_t} (x) > 0 \qquad \text{ for any $R > 0$}. \]
\item[$\bullet$] The condition $\cE_{\text{HF}} (Q) <0$ requires $N = \tr Q$ to be sufficiently large.
\end{itemize}

\begin{proof}
Throughout the proof, we will use the notation $p=-i \nabla_x$. Consider the observable $M = x \sqrt{p^2 + m^2} x = \sum_{j=1}^3 x_j \sqrt{p^2 +m^2} x_j$. We are interested in the time evolution of the expectation $\tr (M Q_t)$.

\medskip

\noindent {\it Step 1.} There exists a constant $C$, only depending on $N (Q)$ and $\cL^2 (Q)$, such that
\begin{equation}\label{eq:step1}
\frac{\rd}{\rd t} \, \tr M Q_t \leq \tr \; (p \cdot x + x \cdot p) \, Q_t + C \, .
\end{equation}

To prove (\ref{eq:step1}) we start by computing
\begin{equation}\label{eq:pf1}
\begin{split}
\frac{\rd}{\rd t} \, \tr \; M Q_t = \; & -i \tr \; M \left[ \sqrt{p^2 + m^2} - \kappa \left( \frac{1}{|.|} * \rho_{Q_t} \right) + \kappa R_{Q_t} , Q_t \right]  \\ = \; & -i \tr
\, \left[ x \sqrt{p^2 + m^2} \, x , \sqrt{p^2 +m^2} \, \right] Q_t \\ &+ i \kappa \tr \, \left[ x \sqrt{p^2 + m^2} \, x, \left( \frac{1}{|.|} * \rho_{Q_t} \right) \right] Q_t \\ & -i \kappa \tr \, \left[ x \sqrt{p^2 + m^2} \, x, R_{Q_t} \right] Q_t \,.
\end{split}
\end{equation}
Using that $x = i \nabla_p$, the first term on the r.h.s. of the last equation is given by
\[ -i \tr \, \left[ x \sqrt{p^2 + m^2} x , \sqrt{p^2 + m^2} \, \right] Q_t = \tr \left( p \cdot x + x \cdot p \right) Q_t \, . \] To control the second term on the r.h.s. of (\ref{eq:pf1}), let $V_{Q_t} (x) = (|.|^{-1} * \rho_{Q_t}) (x)$. Then
\begin{equation*}
\begin{split}
i \kappa \tr \, \Big[ x &\sqrt{p^2 + m^2} \, x, V_{Q_t} \Big] Q_t \\ = \; & i \kappa \tr \, \left( x \sqrt{p^2 + m^2} \, x V_{Q_t} (x) - V_{Q_t} (x) x \sqrt{p^2 + m^2} \, x \right) Q_t \\  = \; &i\kappa \tr \, \left[ \sqrt{p^2 + m^2} , x^2 V_{Q_t} (x) \right] Q_t \\ &- \kappa \tr \, \left( \frac{p}{\sqrt{p^2 + m^2}} \cdot x V_{Q_t} (x) + V_{Q_t} (x) x \cdot \frac{p}{\sqrt{p^2 + m^2}} \right) Q_t \,.
\end{split}
\end{equation*}
Therefore, using Lemma \ref{lm:stein}, we find
\begin{equation}\label{eq:pf2}
\begin{split}
\left| \, i \kappa \tr \, \left[ x \sqrt{p^2 + m^2} \, x, V_{Q_t} \right] Q_t  \right| \leq \; & \kappa N \left\| \left[ \sqrt{p^2 + m^2} , x^2 V_{Q_t} (x) \right] \right\| + \kappa N \| x V_{Q_t} (x) \| \\ \leq \; & \kappa N \left(  \| x V_{Q_t} (x) \| + \| x^2 \nabla V_{Q_t} \| \right)\,.
\end{split}
\end{equation}
By Newton's law (Lemma \ref{lm:nl}), we have
\begin{equation}
\| x V_{Q_t} \| = \sup_{x \in \bR^3} |x| \int \rd y \, \frac{\rho_{Q_t} (y)}{|x-y|} \leq C N
\end{equation}
and
\begin{equation}
\| x^2 \nabla V_{Q_t} \| = \sup_{x \in \bR^3} |x|^2 \left| \nabla_x \int \rd y \, \frac{ \rho_{Q_t} (y)}{|x-y|} \right| \leq C N \, .
\end{equation}
{F}rom (\ref{eq:pf2}), it follows that the second term on the r.h.s. of (\ref{eq:pf1}) is bounded, in absolute value, by
\begin{equation}
\left|\, i \kappa \tr \left[ x \sqrt{p^2 + m^2} \, x, V_{Q_t} \right] Q_t  \right| \leq C\, \kappa N^2 \, .
\end{equation}

Finally, we consider the third term on the r.h.s. of (\ref{eq:pf1}). To this end, we decompose the density $Q_t$ in a sum over orthogonal projections \begin{equation}\label{eq:dec} Q_t = \sum_{j=1}^N |\psi_{j,t} \rangle \langle \psi_{j,t}| \qquad \text{with } \quad \langle \psi_{j,t} , \psi_{i,t} \rangle = \delta_{ij}. \end{equation} It is easy to see that the wave functions $\{ \psi_{j,t} \}_{j=1}^N$ are actually the solution of the Hartree-Fock system (\ref{eq:HF1}) with initial data $\{ \psi_j \}_{j=1}^N$ chosen so that, at time $t=0$, $\langle \psi_i, \psi_j \rangle =\delta_{ij}$ and $Q = \sum_{j=1}^N |\psi_j \rangle \langle \psi_j|$ (it is then easy to show that the r.h.s. of (\ref{eq:dec}) is a solution to (\ref{eq:HF2}); from the local uniqueness of the solution to (\ref{eq:HF2}), we obtain (\ref{eq:dec})). Using this decomposition of the density $Q_t$, we obtain
\begin{equation}\label{eq:pf3}
\begin{split}
\tr \, \Big[x &\sqrt{p^2 + m^2} x, R_{Q_t} \Big] Q_t \\  = \; &\sum_{j=1}^N \langle \psi_{j,t} , \left( \sqrt{p^2 + m^2} x^2 R_{Q_t} - R_{Q_t} x^2 \sqrt{p^2 + m^2} \right) \psi_{j,t} \rangle \\ &+ i \sum_{j=1}^N \langle \psi_{j,t} , \left( \frac{p}{\sqrt{p^2 + m^2}} \cdot x R_{Q_t} + R_{Q_t} x \cdot \frac{p}{\sqrt{p^2 + m^2}} \right) \psi_{j,t} \rangle \\
= \; &\sum_{i,j=1}^N \left\langle \psi_{j,t} , \left[ \sqrt{p^2+m^2} , x^2 \left(\frac{1}{|.|} * \psi_{j,t} \overline{\psi}_{i,t} \right) \right] \psi_{i,t} \right\rangle \\ &+ i \text{Re} \, \sum_{i,j=1}^N \left\langle \psi_{j,t} ,  \frac{p}{\sqrt{p^2 + m^2}} \cdot x \left(\frac{1}{|.|} * \psi_{j,t} \overline{\psi}_{i,t} \right) \psi_{i,t} \right\rangle \, .
\end{split}
\end{equation}
To bound the last term on the r.h.s. of the last equation, we observe that
\begin{equation*}
\begin{split}
\Big| \sum_{i,j=1}^N &\left\langle \psi_{j,t} ,  \frac{p}{\sqrt{p^2 + m^2}} \cdot x \left(\frac{1}{|.|} * \psi_{j,t} \overline{\psi}_{i,t} \right) \psi_{i,t} \right\rangle \Big|  \\ \leq \; & \sum_{i,j=1}^N  \int \rd x \rd y \, \left| \frac{p}{\sqrt{p^2 + m^2}}
\, \psi_{j,t} (x)\right| \,  |\psi_{i,t} (x)| \, |x| \, \frac{|\psi_{j,t} (y)| \, |\psi_{i,t} (y)|}{|x-y|} \\ \leq \; &\int \rd x \rd y \, \rho_{Q_t} (x)\, |x| \, \frac{\rho_{Q_t} (y)}{|x-y|} + \sum_{j}^N  \int \rd x \rd y \, \left| \frac{p}{\sqrt{p^2 + m^2}} \, \psi_{j,t} (x) \right|^2 \, |x| \, \frac{\rho_{Q_t} (y)}{|x-y|}\\ \leq \; &C N^2
\end{split}
\end{equation*}
where we used Newton's law (Lemma \ref{lm:nl}) to perform the $y$-integration, and then, in the second term, we estimated $\tr (p/\sqrt{p^2 +m^2}) Q_t (p/\sqrt{p^2 +m^2}) \leq \tr Q_t = N$.
As for the first term on the r.h.s. of (\ref{eq:pf3}), we can bound its absolute value using Lemma \ref{lm:stein}. We find
\begin{equation}\label{eq:pf4}
\begin{split}
\Big| \sum_{i,j=1}^N \Big\langle \psi_{j,t} , &\left[ \sqrt{p^2+m^2} , x^2 \left(\frac{1}{|.|} * \psi_{j,t} \overline{\psi}_{i,t} \right) \right] \psi_{i,t} \Big\rangle \Big| \\ \leq
\; & \sum_{i,j=1}^N \,
\left\| \left[ \sqrt{p^2+m^2} , x^2 \left(\frac{1}{|.|} * \psi_{j,t} \overline{\psi}_{i,t} \right) \right] \right\| \\ \leq \; &\sum_{i,j=1}^N \left( \left\| x \left( \frac{1}{|.|} * \psi_{j,t} \overline{\psi}_{i,t}\right)\right\| + \left\| x^2 \nabla \left(\frac{1}{|.|} * \psi_{j,t} \overline{\psi}_{i,t} \right) \right\| \right)\,.
\end{split}
\end{equation}
To bound the first contribution, we observe that, by Lemma \ref{lm:nl},
\begin{equation}\label{eq:pf5}
\begin{split}
\left\| x \left( \frac{1}{|.|} * \psi_{j,t} \overline{\psi}_{i,t}\right)\right\| = \; &\sup_{x} |x| \left| \int \rd y \, \frac{1}{|x-y|} \psi_{j,t} (y) \overline{\psi}_{i,t} (y) \right| \\
\leq \; & \sup_x |x| \int \rd y \, \frac{1}{|x-y|} \, \frac{|\psi_{j,t} (y)|^2 + |\psi_{i,t} (y)|^2}{2} \\ \leq \; & \sup_x |x| \int \rd y \, \frac{\rho_{Q_t} (y)}{|x-y|} \leq C N \, .
\end{split}
\end{equation}
Next, we consider the second term on the r.h.s. of (\ref{eq:pf4}). For $x \in \bR^3$, we can find a rotation $R \in SO (3)$ such that $x=R(r e_3)$ where $r = |x|$ and $e_3 = (0,0,1)$. Therefore
\begin{equation}\label{eq:pf6}
\begin{split}
|x|^2 \left| \int \rd y \, \frac{(x-y)}{|x-y|^3} \, \psi_{j,t} (y) \overline{\psi}_{i,t} (y) \right| = \; & \left| \, r^2 \int \rd y \frac{(R(r e_3) - y)}{|R(re_3) - y|^3} \, \psi_{j,t} (y) \, \overline{\psi}_{i,t} (y) \right| \\
\leq \; & \sum_{i=1}^2 \left| \, r^2 \int \rd y \frac{y_i}{|re_3 - y|^3} \, \psi_{j,t} (Ry) \, \overline{\psi}_{i,t} (Ry) \right| \\ &+ \left| \, r^2 \int \rd y \frac{r-y_3}{|re_3 - y|^3} \, \psi_{j,t} (Ry) \, \overline{\psi}_{i,t} (Ry) \right|\,.
\end{split}
\end{equation}
The last term on the r.h.s. of (\ref{eq:pf6}) can be estimated by
\begin{equation}
\begin{split}
\left| \, r^2 \int \rd y \frac{r-y_3}{|re_3 - y|^3} \, \psi_{j,t} (Ry) \, \overline{\psi}_{i,t} (Ry) \right| \leq \; &  r^2 \int \rd y \frac{|r-y_3|}{|re_3 - y|^3} \, |\psi_{j,t} (Ry)| \, |\psi_{i,t} (Ry)| \\ \leq \; &  r^2 \int \rd y \frac{|r-y_3|}{|re_3 - y|^3} \, \rho_t (Ry)  \\
\leq \; &  r^2 \int \rd y \frac{|r-y_3|}{|r - y|^3} \, \rho_t (y)
\end{split}
\end{equation}
Introducing spherical coordinates for $y = s \hat {y}$ for $\hat{y} \in S^2$, we find
\begin{equation}\label{eq:pf7}
\begin{split}
\left| \, r^2 \int \rd y \frac{r-y_3}{|re_3 - y|^3} \, \psi_{j,t} (sy) \, \overline{\psi}_{i,t} (s y) \right| \leq \; & \int_0^{\infty}
\rd s s^2 \, \rho_t (s) \int_{S^2} \rd \hat{y} \, \frac{|1 - \frac{s \hat{y}_3}{r}|}{|e_3 - \frac{s\hat{y}}{r}|} \leq C N \, .
\end{split}
\end{equation}
Here we used the fact that, as we prove in Lemma \ref{lm:intbd} below,
\begin{equation}\label{eq:int} \sup_{\lambda >0} \int_{S^2} \rd \hat{y} \frac{|1 - \lambda \hat{y}_3|}{|e_3 - \lambda \hat{y}|} < \infty \,. \end{equation}
As for the first term on the r.h.s. of (\ref{eq:pf6}), we remark that, for example, the summand with $i=1$ can be controlled as follows.
\begin{equation}\label{eq:pf8}
\begin{split}
\left| \, r^2 \int \rd y \frac{y_1}{|re_3 - y|^3} \psi_{j,t} (Ry)
\overline{\psi}_{i,t} (Ry) \right| \leq \;&
\left| \, r \int \rd y  \, \frac{(r-y_3)y_1}{|re_3 - y|^3} \psi_{j,t} (Ry)
\overline{\psi}_{i,t} (Ry) \right|  \\ &+
\left| \, r \int \rd y \, y_3 \frac{y_1}{|re_3 - y|^3} \psi_{j,t} (Ry)
\overline{\psi}_{i,t} (Ry) \right| \, .
\end{split}
\end{equation}
Observing that
\begin{equation*}
\begin{split}
\Big| \, r \int \rd y \, y_3 \, \frac{y_1}{|re_3 - y|^3} \, &\psi_{j,t} (Ry) \,
\overline{\psi}_{i,t} (Ry) \Big| \\ = \; & \left|\, r \int \rd y \, y_3 \, \partial_{y_1} \frac{1}{|re_3 - y|} \, \psi_{j,t} (Ry) \, \overline{\psi}_{i,t} (Ry) \right| \\
\leq \; &  \left| \, r \int \rd y \, \left( y_3 \partial_{y_1} - y_1 \partial_{y_3} \right) \frac{1}{|re_3 - y|} \psi_{j,t} (Ry) \overline{\psi}_{i,t} (Ry) \right| \\ &+
r \int \rd y \, |y_1| \frac{|r-y_3|}{|re_3 - y|^3} |\psi_{j,t} (Ry)|  |\psi_{i,t} (Ry)|
\end{split}
\end{equation*}
it follows from (\ref{eq:pf8}) that (using the notation $\ph_{j,t} (y) = \psi_{j,t} (Ry)$)
\begin{equation}\label{eq:pf8b}
\begin{split}
\Big| \, r^2 &\int \rd y \frac{y_1}{|re_3 - y|^3} \psi_{j,t} (Ry)
\overline{\psi}_{i,t} (Ry) \Big| \\ \leq \; &\left| \, r \int \rd y \, \frac{1}{|re_3 - y|} \left( L_2 \, \ph_{j,t} \right)(y) \overline{\ph}_{i,t} (y) \right| +\left| \, r \int \rd y \, \frac{1}{|re_3 - y|} \left( L_2 \overline{\ph}_{i,t} \right) (y) \ph_{j,t} (y) \right| \\ &+ r \int \rd y \,  \frac{1}{|re_3 - y|^3} |\psi_{j,t} (Ry)|  |\psi_{i,t} (Ry)|
\\ \leq \; & r \int \rd y \, \frac{1}{|re_3 - y|} \left( |L_2 \ph_{j,t}(y)| |\ph_{i,t} (y)| + | L_2 \ph_{i,t} (y)| |\ph_{j,t} (y)| \right) + r \int \rd y \,  \frac{\rho_t (y)}{|re_3 - y|^3}.
\end{split}
\end{equation}
The last term is bounded by $C N$ by Lemma \ref{lm:nl}. The first term on the r.h.s. of the last equation, can be controlled by
\begin{equation}\label{eq:pf9}
\begin{split}
r \int \rd y \, \frac{1}{|re_3 - y|} \, &\left(
|L_2 \, \ph_{j,t}(y)| |\ph_{i,t} (y)| + | L_2 \ph_{i,t} (y)| |\ph_{j,t} (y)| \right) \\ &\leq  r \sum_{j=1}^N \int \rd y \, \frac{1}{|re_3 -y|} \left( |L_2 \, \ph_{j,t}(y)|^2 + |\ph_{j,t} (y)|^2 \right).
\end{split}
\end{equation}
Note that $\sum_{j=1}^N  |\ph_{j,t} (y)|^2 = \rho_{Q_t} (Ry) = \rho_{Q_t} (y)$. Moreover, we have
\[ \sum_{j=1}^N |L_2 \ph_{j,t} (y)|^2 = \sum_{j=1}^N \left(L_2 |\ph_{j,t} \rangle \langle \ph_{j,t}| L_2 \right) (y,y) = (L_2 \widetilde Q_t L_2)(y,y) \] where we defined $\widetilde Q_t = \sum_j |\ph_{j,t} \rangle \langle \ph_{j,t}|$. Since
\[ \widetilde Q_t (x,y) = \sum_j \ph_{j,t} (x) \overline{\ph}_{j,t} (y) = \sum_j \psi_{j,t} (Rx) \overline{\psi}_{j,t} (Ry) = Q_t (Rx,Ry) = Q_t (x,y) \] for all $x,y \in \bR^3$, it follows that $\widetilde Q_t = Q_t$. Therefore, from (\ref{eq:pf8b}) and (\ref{eq:pf9}) we conclude that
\begin{equation*}
\begin{split}
\left| \, r^2 \int \rd y \frac{y_1}{|re_3 - y|^3} \psi_{j,t} (Ry)
\overline{\psi}_{i,t} (Ry) \right| \leq  \; & C N +  r \int \rd y \, \frac{1}{|re_3 -y|}  (L_2 Q_t L_2 ) (y,y)  \\ \leq \; & C N + \sum_{\ell=1}^3 r \int \rd y \, \frac{1}{|re_3 -y|}  (L_{\ell} Q_t L_{\ell} ) (y,y).
\end{split}
\end{equation*}
Since $\sum_{\ell=1}^3 (L_{\ell} Q_t L_{\ell}) (y,y)$ is invariant w.r.t. rotations of $y$, we can apply Lemma \ref{lm:nl} to obtain that
\begin{equation*}
\begin{split}
\left| \, r^2 \int \rd y \frac{y_1}{|re_3 - y|^3} \psi_{j,t} (Ry)
\overline{\psi}_{i,t} (Ry) \right| \leq \; &C N +
 C \sum_{\ell=1}^3 \int \rd y \,   (L_{\ell} Q_t L_{\ell} ) (y,y) \\ = \; & C N + C \tr \; L^2 Q_t \, .
\end{split}
\end{equation*}
Using Lemma \ref{lm:L2} and inserting the last bound and (\ref{eq:pf7})
back in (\ref{eq:pf6}) we find that
\begin{equation*}
|x|^2 \Big| \int \rd y \frac{(x-y)}{|x-y|^3} \psi_{j,t} (y) \overline{\psi}_{i,t} (y) \Big| \leq C N + C \tr L^2 Q
\end{equation*}
and thus, from (\ref{eq:pf4}), that
\begin{equation*}
\begin{split}
\Big| \sum_{i,j=1}^N &\left\langle \psi_{j,t} , \left[ \sqrt{p^2+m^2} , x^2 \left(\frac{1}{|.|} * \psi_{j,t} \overline{\psi}_{i,t} \right) \right] \psi_{i,t} \right\rangle \Big| \leq C N^3 (Q) + C N^2 (Q) \cL^2 (Q)
\end{split}
\end{equation*}
{F}rom (\ref{eq:pf1}), we obtain (\ref{eq:step1}).

\medskip

\noindent {\it Step 2.} If $Q_t$ is a solution of the Hartree-Fock equation (\ref{eq:HF2}),
we have
\begin{equation}
\label{eq:step2}
\frac{\rd}{\rd t} \, \tr \; (p\cdot x + x \cdot p) \, Q_t \leq 2 \, \cE_{\text{HF}} (Q) .
\end{equation}

To show (\ref{eq:step2}), we compute
\begin{equation}\label{eq:td}
\begin{split}
\frac{\rd}{\rd t} \, \tr \; &(p \cdot x + x \cdot p ) \, Q_t \\ = &\; -i \tr \; (p\cdot x + x \cdot p) \left[ \sqrt{p^2 + m^2} - \kappa \left( \frac{1}{|.|} * \rho_{Q_t} \right) + \kappa R_{Q_t} , Q_t \right] \\ = & \; -i \tr \left[ (p\cdot x + x \cdot p), \sqrt{p^2 + m^2} - \kappa \left( \frac{1}{|.|} * \rho_{Q_t} \right) + \kappa R_{Q_t} \right] \, Q_t \,.
\end{split}
\end{equation}
Now we observe that
\begin{equation}\label{eq:td1}
-i \tr \; \left[ (p\cdot x + x \cdot p), \sqrt{p^2 + m^2} \, \right] Q_t \, = 2 \tr \;  \frac{p^2}{\sqrt{p^2 + m^2}} \, Q_t \leq 2 \tr \, \sqrt{p^2 + m^2} \, Q_t
\end{equation}
and
\begin{equation}\label{eq:td2}
\begin{split}
i\kappa \tr \; \left[ (p\cdot x + x \cdot p), \left( \frac{1}{|.|} * \rho_{Q_t} \right) \right] \, Q_t = \; & 2 \kappa \tr \; x \cdot \nabla \left( \frac{1}{|.|} * \rho_{Q_t} \right) Q_t  \,
\\ = \; &-2 \kappa \int \rd x \, \rd y \; x \cdot \frac{(x-y)}{|x-y|^3} \, \rho_{Q_t} (y) \, \rho_{Q_t} (x) \\ = \; & 2 \kappa \int \rd x \, \rd y \; y \cdot \frac{(x-y)}{|x-y|^3} \, \rho_{Q_t} (y) \, \rho_{Q_t} (x) \\ = \; &\kappa \int \rd x \, \rd y \, \frac{1}{|x-y|} \, \rho_{Q_t} (y) \, \rho_{Q_t} (x)\,.
\end{split}
\end{equation}
As for the contribution to (\ref{eq:td}) from the term with $R_{Q_t}$, we have
\begin{equation}\label{eq:td3}
\begin{split}
-i \kappa \; \tr \; \Big[ \, &( x \cdot p + p \cdot x) , R_{Q_t} \, \Big] \, Q_t \\ = \; &
-2i \kappa \; \tr \; \left[ \, x \cdot p \, , \, R_{Q_t} \, \right] \, Q_t \\ = \;
&-2i \kappa \; \tr \; \left( x \cdot p \, R_{Q_t} \, Q_t - R_{Q_t} \, x \cdot p \, Q_t \right) \\ = \; &
-2 \kappa \int \rd x \, \rd y \; \left( x \cdot \nabla_x \left(\frac{Q_t (x,y)}{|x-y|} \right) Q_t (y,x) - Q_t (y,x) x \cdot \nabla_x Q_t (x,y) \right) \\
=\; & 2 \kappa \int \rd x \, \rd y \; x \cdot \frac{x-y}{|x-y|^3} |Q_t (x,y)|^2 \\
= \; & \kappa \int \rd x \, \rd y \; \frac{|Q_t (x,y)|^2}{|x-y|}\,.
\end{split}
\end{equation}
{F}rom (\ref{eq:td}), (\ref{eq:td1}), (\ref{eq:td2}), and (\ref{eq:td3}), we obtain (\ref{eq:step2}).

\medskip

\noindent{\it Step 3.} If $Q_t$ is a spherically symmetric solution to (\ref{eq:HF2}), there exists a constant $C$, only depending from $N(Q)$ and $\cL^2 (Q)$ such that
\begin{equation}\label{eq:step3}
\tr \; M Q_t \leq t^2 \cE_{\text{HF}} (Q) + t \, \left(\tr \, (x \cdot p + p \cdot x) Q  + C \right) + \tr \; M Q \, .
\end{equation}

Eq. (\ref{eq:step3}) follows directly from the statements proven in Step 1 and Step 2, integrating twice over time.

\medskip

\noindent{\it Step 4. Conclusion of the proof.} {F}rom the assumption (\ref{eq:tech}) on the initial density $Q$, it follows immediately that \[
\tr \, M Q = \tr x \sqrt{p^2 + m^2} \, x \, Q \leq \tr \, (1 + x^4 + p^2) Q < \infty \] and that \[ \left| \tr \, (x \cdot p + p \cdot x) Q \right| \leq \tr \, (x^2 + p^2) Q < \infty\, . \]
Thus, if $\cE_{\text{HF}} (Q) <0$, (\ref{eq:step3}) contradicts, for $t$ large enough, the non-negativity of the expectation $\tr M Q_t$. This implies immediately that the maximal existence time $T$ for the local solution $Q_t$ is finite. {F}rom the blow-up alternative (see Theorem~\ref{thm:loc}), it follows that there exists $T < \infty$ with $\| Q_t \|_{\cH^{1/2}} = \tr \, (1-\Delta)^{1/2} \, Q_t \to \infty$ as $t \to T^-$.
\end{proof}

In the proof of Theorem \ref{thm:main}, we had to control commutators of the pseudo-differential operator $\sqrt{p^2 + m^2}$ with multiplication operators of the form $f(x)$ (see for example (\ref{eq:pf4})). In this respect, it turns out that the Calderon-Zygmund theory of singular integrals is very useful.
\begin{lemma}\label{lm:stein}
Suppose $m >0$, $p = -i\nabla$. Then, for every $f \in W^{1,\infty} (\bR^3)$, we have
\begin{equation}
\left\| \left[ \sqrt{p^2 + m^2} , f(x) \right] \right\| \leq C \| \nabla f \|_{\infty} \, .
\end{equation}
\end{lemma}
A proof of this lemma can be found in \cite{Stein}; see in particular the Corollary on page~309. The statement of this corollary does not give an effective bound on the norm of the commutator. However, the corollary is based on Theorem 3, on page~294 of \cite{Stein}, whose proof provides the effective control we need (a remark in this sense can be found in the paragraph 3.3.5, on page 305 of \cite{Stein}).

\medskip

Finally, in the next lemma we give a proof of the bound (\ref{eq:int}).
\begin{lemma}\label{lm:intbd} We have
\begin{equation}
\sup_{\lambda \geq 0} \int_{S^2} \rd \hat{y} \frac{|1-\lambda \hat{y}_3|}{|e_3 - \lambda \hat{y} |^3}  < \infty \, .
\end{equation}
\end{lemma}

\begin{proof}
First, we observe that
\begin{equation}
\sup_{\lambda < 1/2, \lambda > \, 2} \int_{S^2} \rd y \frac{|1-\lambda y_3|}{|e_3 - \lambda y |^3} \leq C
\end{equation}
because, in this regime of $\lambda$, there is no singularity from the denominator. On the other hand, for arbitrary $\lambda \in [1/2, 2]$, we have
\begin{equation*}
\begin{split}
\int_{S^2} \rd \hat{y} \frac{|1-\lambda \hat{y}_3|}{|e_3 - \lambda \hat{y} |^3} = \; & 2 \pi \int_0^{\pi} \rd \theta \, \sin \theta \, \frac{|1-\lambda\cos \theta|}{\left( \lambda^2 \sin^2 \theta + (1-\lambda \cos \theta)^2 \right)^{3/2} } \\
=\; & 2\pi \int_0^{\pi} \rd \theta \, \sin \theta \, \frac{|1-\lambda\cos \theta|}{\left( \frac{\lambda^2-1}{2} + 1 - \lambda \cos \theta\right)^{3/2}} \\ = \; & \frac{2\pi}{\lambda} \int_{1-\lambda}^{1+\lambda} \rd z \, \frac{|z|}{\left(\frac{\lambda^2 -1}{2} + z \right)^{3/2}} =  \frac{2\pi}{\lambda} \int_0^{2\lambda} \rd x \, \frac{|x+(1-\lambda)|}{\left(\frac{(\lambda -1)^2}{2} + x \right)^{3/2}}
\end{split}
\end{equation*}
Therefore
\begin{equation*}
\begin{split}
\int_{S^2} \rd \hat{y} \frac{|1-\lambda \hat{y}_3|}{|e_3 - \lambda \hat{y} |^3}
\leq\; &C \int_0^{2\lambda} \rd x \, \frac{|x|}{ \left(\frac{(\lambda -1)^2}{2} + x \right)^{3/2}} +  C \int_0^{2\lambda} \rd x \, \frac{(1-\lambda)}{ \left(\frac{(\lambda -1)^2}{2} + x \right)^{3/2}}
\\ \leq \; & C \int_0^4 \frac{\rd x}{|x|^{1/2} } + C \int_0^{\infty} \rd x \, \frac{1}{ \left(1 + x \right)^{3/2}} \leq C \,.
\end{split}
\end{equation*}
\end{proof}

\section*{Acknowledgments}
The authors would like to thank Lei Zhang and Enno Lenzmann for very helpful remarks. B. S. is on leave from the University of Cambridge, UK. His research is supported by a Sofja Kovalevskaja Award of the Alexander von Humboldt Foundation.

\bibliographystyle{amsplain}

\end{document}